\documentclass[11pt]{article}
\usepackage{fullpage}

\usepackage[utf8]{inputenc}
\usepackage{setspace}
\usepackage[margin=1.25in]{geometry}
\usepackage{graphicx}
\graphicspath{ {./figures/} }
\usepackage{subcaption}
\usepackage{amsmath}
\usepackage{lineno}
\usepackage{amsthm}
\usepackage{color,soul}
\usepackage{comment}
\usepackage{hyperref}
\usepackage{nicematrix}

\usepackage{cleveref}

\usepackage[section]{placeins}

\newcommand{\ignore}[1]{}

\usepackage{complexity}
\usepackage{tikz}
\usepackage{gensymb}
\usetikzlibrary{shapes.geometric}
\usetikzlibrary{shapes}
\usepackage{concmath,ccfonts,beton,euler}

\usetikzlibrary{positioning}
\usetikzlibrary{decorations.text}
\tikzset{
    ncbar angle/.initial=90,
    ncbar/.style={
        to path=(\tikztostart)
        -- ($(\tikztostart)!#1!\pgfkeysvalueof{/tikz/ncbar angle}:(\tikztotarget)$)
        -- ($(\tikztotarget)!($(\tikztostart)!#1!\pgfkeysvalueof{/tikz/ncbar angle}:(\tikztotarget)$)!\pgfkeysvalueof{/tikz/ncbar angle}:(\tikztostart)$)
        -- (\tikztotarget)
    },
    ncbar/.default=0.5cm,
}
\tikzset{round left paren/.style={ncbar=0.5cm,out=100,in=-100}}
\tikzset{round right paren/.style={ncbar=0.5cm,out=80,in=-80}}

\usepackage{amsmath}
\usepackage{amssymb}
\usepackage{array}

\newtheorem{theorem}{Theorem}[section]
\newtheorem{lemma}[theorem]{Lemma}
\newtheorem{corollary}[theorem]{Corollary}

\newtheorem{definition}[theorem]{Definition}

\usepackage{subcaption}
\usepackage{autonum}
 
\newcommand{\set}[2]{\{\,#1\mid#2\,\}}

\newcommand{\kltight}{$(k,l)$-tight}

\newcommand{\klspars}{$(k,l)$-sparse}

\newcommand{\ffact}{$f$-factor}

\DeclareMathOperator{\Logspace}{\mathsf{Logspace}}

\pgfdeclarelayer{background}
\pgfsetlayers{background,main}
\def\drawpolygon#1,#2;{
    \begin{pgfonlayer}{background}
        \filldraw[line width=20,join=round      ](#1.center)foreach\A in{#2}{--(\A.center)}--cycle;
        \filldraw[line width=19,join=round,white](#1.center)foreach\A in{#2}{--(\A.center)}--cycle;
    \end{pgfonlayer}
}

\title{Planarizing Gadgets for $(k,l)$-tight Graphs Do Not Exist\thanks{Partially supported by DFG grant TH 472/5-2
and SERB MATRICS grant MTR2022/001009
}} 

\author{Archit Chauhan\thanks{IIT Bombay} \and Rohit Gurjar\thanks{IIT Bombay} \and Kilian Rothmund\thanks{Ulm University, Aalen University} \and Thomas Thierauf\thanks{Ulm University}}

\begin{document}
\maketitle

\setstcolor{red} 

\begin{abstract}
The problem of \emph{recognizing $(k,l)$-tight graphs} is a fundamental problem that has close connections to well studied problems
like graph rigidity. The problem is better understood for planar graphs as compared to general graphs. For example, deterministic 
\textsf{NC}-algorithms for the problem are known for planar graphs, but no such algorithm is known for general graphs.
A common approach to reduce a graph problem to the planar case is to use \emph{planarizing gadgets}.
Our main contribution is to show that, unconditionally, planarizing gadgets for the problem of \emph{recognizing $(k,l)$-tight graphs} do not exist.
\end{abstract}

\section{Introduction}
Let $G = (V,E)$ be a (multi-) graph with $|V| = n$ vertices and $|E| = m$ edges.
Graph~$G$ is \emph{$(k,l)$-sparse},
for integers $k,l$, 
if every subset $S \subseteq V$,
where $|S| \geq 2$,
induces a subgraph in~$G$ with at most $k|S| - l$
many edges.
A graph~$G$ is \emph{$(k,l)$-tight}, if~$G$ is $(k,l)$-sparse and 
$m = kn-l$.
Graph~$G$ is $(k,l)$-\emph{spanning}, if it contains a $(k,l)$-tight subgraph that spans the entire vertex set.

The class of $(k,l)$-sparse/tight/spanning graphs captures several interesting graph classes. 
We give some examples.
\begin{itemize}
\item 
The $(2,3)$-tight graphs are the \emph{minimally rigid graphs} 
in~$2D$, also called \emph{Laman graphs}.
The $(2,3)$-spanning graphs are the \emph{rigid graphs} in~$2D$~\cite{Laman1970, PollaczekGeiringer1927}.
\item
The $(k,k)$-tight graphs are the graphs where the edges
can be partitioned into~$k$ edge-disjoint spanning 
trees.
In case of $(k,k)$-sparse graphs,
the edges can be partitioned into~$k$ edge-disjoint 
forests~\cite{NashWilliams1961, Tutte1961}.
\item 
For $k \leq l < 2k$, the $(k,l)$-sparse graphs 
are precisely the $(k,l-k)$-arborescences~\cite{haas02}.
A $(k,a)$-\emph{arborescence} is a graph where adding any $a$ edges results in a graph whose edge set is the union of $k$ edge-disjoint spanning trees.
\end{itemize}

A $(k,l)$-spanning graph $G = (V,E)$ defines a \emph{matroid}
for $l < 2k$ and large enough~$n$~\cite{LEE2008}:
The ground set is~$E$.
A set~$I \subseteq E$ is independent,  if $(V,I)$ is $(k,l)$-sparse.
The base sets are the $(k,l)$-tight subsets.

\ignore{
Another aspect is that for integers~$k$ and $l\leq 2k-1$, the collection of all $(k,l)$-tight graphs (with sufficiently large number of vertices) is the set of bases of a matroid where the ground set are the egdes of the complete multi-graph with loop multiplicity $k-l$ and edge multiplicity $2k-l$~\cite{LEE2008}.
}

For rational parameters $k,l$, there are also applications in discrete geometry~\cite{Whiteley1996}.


\paragraph{Complexity.}
Whether a graph $G = (V,E)$ is
sparse, tight, or spanning for parameters $(k,l)$
can be decided in time $\mathcal{O}(n^2)$~\cite{LEE2008},
for $|V| = n$ and $l\leq 2k-1$.

With respect to the \emph{parallel complexity},
it is open whether $(k,l)$-sparse, tight, or spanning graphs can be recognized  in~$\NC$, or even~$\Logspace$. 
Note that the problem of deciding if a graph is \klspars{} is slightly more general than deciding if a graph is \kltight
\footnote{To see this, observe that if we are given an oracle to decide $(k,l)$-sparsity, then for any graph, we can decide if it is $(k,l)$-tight by just 
checking if $m=kn-l$. However if we have a graph where $m\neq kn-l$, then even with an oracle for deciding $(k,l)$-tightness there is no 
obvious efficient way to rule out $(k,l)$-sparsity.}.
Hendrickson observed that the
decision problem for $(2,3)$-sparse graphs reduces to 
\emph{bipartite perfect matching}~\cite{Hendrickson1992},
and this reduction works similar for $(k,l)$-sparse graphs.
Since bipartite perfect matching is in 
$\RNC$~\cite{Mulmuley1987bipMatchingRNC}
and in quasi-$\NC$~\cite{Fenner2016},
we get the same bounds, $\RNC$ and quasi-$\NC$, 
for deciding $(k,l)$-sparsity.

For \emph{planar graphs}, we get $\NC$-algorithms in some cases.
For example, whether a planar graph~$G$ is $(2,3)$-tight
can be decided in~$\NC$,
in fact,
even when~$G$ is only \emph{one-crossing-minor-free}~\cite{grt25}.

Hendrickson's reduction from $(k,l)$-sparsity to bipartite perfect matching does not preserve planarity.
However, there is a reduction from $(k,l)$-sparsity to $s$-$t$ \emph{max flow} (see~\cite{HLS97,HLS98}), that can be tweaked to reduce $(k,l)$-sparsity to \emph{max flow with multiple sources and sinks} such that  planarity is preserved.
For completeness, we give this reduction in the appendix Section~\ref{app:klsparse_to_flow}.
Sankowski~\cite{PS18} reduces max flow to \emph{weighted perfect matching} which can be solved in $\NC$ when the input graph is planar or has bounded genus~\cite{Anari2020}.
Sankowski's reduction preserves planarity and also the genus of the graph. Therefore, we can decide $(k,l)$-sparsity in $\NC$ for planar graphs and also for bounded-genus graphs.
Moroever, max flow is in $\NC$ when the input graph has bounded treewidth~\cite{HAGERUP1998366}.
Hence, $(k,l)$-sparsity is also in $\NC$ for bounded-treewidth graphs.
We describe this in more detail in Appendix~\ref{app:klsparse_to_flow}.

For recognizing planar $(k,l)$-tight graphs, there is also an alternative $\NC$ algorithm. For tightness one can modify the flow reduction by fixing \emph{demands} at the source and sink vertices.
Miller and Naor~\cite{Miller1995} solve this more restricted flow problem in~$\NC$ when the input graph is planar.

Note that it is open whether $(k,l)$-spanning graphs can be recognized in~$\NC$, even for planar graphs.

\paragraph{Our contribution.}
By the results for planar graphs mentioned above,
an obvious idea to bring the decision problem for
$(k,l)$-tight graphs down to~$\NC$ is
to reduce it to the planar case.
A common approach to such a reduction is a \emph{planarizing gadget} that locally replaces edge crossings from a given drawing of a graph while preserving a desired property, like $(k,l)$-sparsity.
Planarazing gadgets have been used to show that classical $\NP$-hard graph problems, like $3$-colorability and vertex cover, are still $\NP$-hard in the planar case~\cite{GAREY1976237}. 
Datta, Kulkarni, Limaye and Mahajan used planarizing gadgets in the context of computing the determinant of an adjacency matrix~\cite{datta07}.
We show unconditionally,
that such planarizing gadgets do not exist for $(k,l)$-tight graphs.

\begin{theorem}\label{thm:tightGadget}
Planarizing gadgets do not exist for
$(k,l)$-tight graphs,
for $k \geq 2$ and $0 \leq l \leq 2k - 1$.
\end{theorem}

The result from Theorem~\ref{thm:tightGadget} holds analogously for \emph{bipartite perfect matching}~\cite{matchingGadget}.
The question whether bipartite perfect matching is 
in $\NC$ is still open and much work has been done to consider
different variations of the problem in restricted graph classes or other complexity classes closely related to  
$\NC$~\cite{Fenner2016, Mulmuley1987bipMatchingRNC, raghav06, Datta2009, Anari2020, Mahajan2000, ola2017, AM25, PS18}.
As discussed above, recognizing $(k,l)$-tight graphs reduces to bipartite perfect matching. A reduction in the other direction is not known.
Hence, $(k,l)$-tightness might be easier to solve.
Though this problem reduces to bipartite perfect matching, not all results from bipartite perfect matching carry over to $(k,l)$-tight graphs directly.
In particular, the unconditional non-existence result for planarizing gadgets for bipartite perfect matching by Gurjar, Messner, Straub and Thierauf~\cite{matchingGadget} does not seem to carry over to $(k,l)$-tight graphs and new ideas are required to solve the problem.

Also in terms of sequential time complexity our result might be interesting: For recognizing $(2,3)$-tight graphs, the best known sequential algorithm has a running time of $\mathcal{O}(n\sqrt{n \log n} + m)$~\cite{Gabow1992}. Planar $(2,3)$-tight graphs (a.k.a. \emph{planar Laman graphs}) have a geometric characterization via \emph{pointed pseudo-triangulations}~\cite{streinu2000, haas2005}.
Using this geometric characterization, planar $(2,3)$-tight graphs can be recognized in $\mathcal{O}(n \log^3 n)$, thereby improving the time complexity for the general case~\cite{rss}. 
Our non-existence result for planarizing gadgets poses an obstacle in improving the sequential running time for the general case by reducing it to the planar case.


\section{Preliminaries}\label{s:prel}
Let $G = (V,E)$ be a graph.
For a set of vertices $S \subseteq V$, we denote the subgraph of~$G$ induced by~$S$ with $G[S]$. 
We denote edges in the graph~$G[S]$ by~$E[S]$, i.e.
\[
E[S] = \set{(u,v) \in E}{u,v \in S}.
\]
\ignore{For two subsets $X,Y \subseteq V$ we denote the edges in~$G$ between $X,Y$ by $E_G(X,Y)$, i.e.
\[
E_G(X,Y) = \set{(u,v) \in E}{u \in X, v \in Y}.
\]}
If the edge set of a graph~$G$ is not explicitly given a name,
we denote it by~$E_G$.

We use integers~$k$ and~$l$ for sparsity and tightness
conditions of graphs throughout the paper.

\subsection{Planarizing gadgets} 
Roughly, a \emph{planarizing gadget} is a planar graph of constant size which is used to replace edge crossings in the drawing of any graph $G$ like in Figure~\ref{fig:gadget}, while preserving a desired property.
More formally, let $\Gamma$ be a graph of constant size that has a planar embedding with at least four vertices $a,b,c,d$ that appear on the outer face in the order $a,c,b,d$.
Let $G$ be a graph with a drawing in the plane 
where edges $ab, cd$ are crossing.
When we replace this crossing in the drawing of $G$ by $\Gamma$ as in Figure~\ref{fig:gadget}, we get a drawing of $G' = G - \{ab, cd\} + \Gamma$ where the number of crossings is decreased by one.
The idea is to replace each crossing in $G$ by a copy of $\Gamma$ to make the graph planar.
Let $L$ be a class of graphs.
Then, the graph~$\Gamma$ is called a \emph{planarizing gadget for $L$} if we have
\[
G \in L \Leftrightarrow G' \in L.
\]
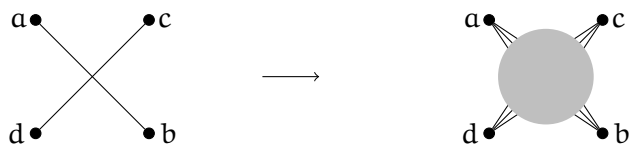
\begin{figure}[bth]
\centering
    \usetikzlibrary{positioning, fit, shapes.geometric}

\tikzstyle{every node}=[circle, draw, fill=black,inner sep=0pt, minimum width=4pt,
node distance =1 cm and 1cm ] 
\begin{tikzpicture}[scale=1.5]
    \node (a) at (0,0) [label=left:$a$]{}; 
    \node (c) at (1,0) [label=right:$c$]{};
    \node (d) at (0,-1) [label=left:$d$]{};
    \node (b) at (1,-1) [label=right:$b$]{};

	\draw[] (a) -- (b);
	\draw[] (c) -- (d);
    \draw[->] (2,-0.5) -- (2.5,-0.5);

    \node (a1) at (4,0) [label=left:$a$]{};
    \node (c1) at (5,0) [label=right:$c$]{};
    \node (d1) at (4,-1) [label=left:$d$]{};
    \node (b1) at (5,-1) [label=right:$b$]{};

    \node (in1) at (4.2,-0.3) [fill=white, draw=none]{};
    \node (in2) at (4.25,-0.25) [fill=white, draw=none]{};
    \node (in3) at (4.3,-0.2) [fill=white, draw=none]{};
    \draw (a1) -- (in1);
    \draw (a1) -- (in2);
    \draw (a1) -- (in3);

    \node (in4) at (4.8,-0.3) [fill=white, draw=none]{};
    \node (in5) at (4.75,-0.25) [fill=white, draw=none]{};
    \node (in6) at (4.7,-0.2) [fill=white, draw=none]{};
    \draw (c1) -- (in4);
    \draw (c1) -- (in5);
    \draw (c1) -- (in6);

    \node (in7) at (4.2,-0.7) [fill=white, draw=none]{};
    \node (in8) at (4.25,-0.75) [fill=white, draw=none]{};
    \node (in9) at (4.3,-0.8) [fill=white, draw=none]{};
    \draw (d1) -- (in7);
    \draw (d1) -- (in8);
    \draw (d1) -- (in9);

    \node (in10) at (4.8,-0.7) [fill=white, draw=none]{};
    \node (in11) at (4.75,-0.75) [fill=white, draw=none]{};
    \node (in12) at (4.7,-0.8) [fill=white, draw=none]{};
    \draw (b1) -- (in10);
    \draw (b1) -- (in11);
    \draw (b1) -- (in12);

    \draw[fill=lightgray, draw=none] (4.5,-0.5) circle (12pt);

 \end{tikzpicture}
        \caption{The crossing $ab, cd$ is resolved by adding a planarizing gadget}
        \label{fig:gadget}
\end{figure}
Note that there might be multiple crossings in a drawing of $G$. In particular, one edge could cross with multiple edges.
We replace each crossing by a copy of $\Gamma$.
When there are multiple crossings, one might need a more restricted gadget like the one in Figure~\ref{fig:multCrossingGadget}.
However, when we show that a planarizing gadget does not exist for the case that we only consider one edge crossing in the drawing, more restricted gadgets for multiple crossings cannot exist as well. 
Throughout the paper we use $G'$ to denote the graph $G - \{ab, cd\} + \Gamma$, i.e. the graph obtained by replacing the crossing edges $\{ab,cd\}$ with the gadget.
\begin{figure}[bth]
\centering
    \usetikzlibrary{positioning, fit, shapes.geometric}

\tikzstyle{every node}=[circle, draw, fill=black,inner sep=0pt, minimum width=4pt,
node distance =1 cm and 1cm ] 
\begin{tikzpicture}[scale=1.5]
    \node (a) at (0,0) [label=left:$a$]{}; 
    \node (c) at (1,0) [label=right:$c$]{};
    \node (d) at (0,-1) [label=left:$d$]{};
    \node (b) at (1,-1) [label=right:$b$]{};

	\draw[] (a) -- (b);
	\draw[] (c) -- (d);
    \draw[->] (2,-0.5) -- (2.5,-0.5);

    \node (a1) at (4,0) [label=left:$a$]{};
    \node (c1) at (5,0) [label=right:$c$]{};
    \node (d1) at (4,-1) [label=left:$d$]{};
    \node (b1) at (5,-1) [label=right:$b$]{};

    \draw[fill=lightgray, draw=none] (4.25,-0.75) rectangle (4.75,-0.25);
    
    \node (a11) at (4.25,-0.25) []{};
    \node (c11) at (4.75,-0.25) []{};
    \node (d11) at (4.25,-0.75) []{};
    \node (b11) at (4.75,-0.75) []{};

	\draw[] (a1) -- (a11);
    \draw[] (b1) -- (b11);
    \draw[] (c1) -- (c11);
    \draw[] (d1) -- (d11);

 \end{tikzpicture}
        \caption{The crossing $ab, cd$ is resolved by adding a more restricted gadget that also works for multiple crossings.}
        \label{fig:multCrossingGadget}
\end{figure}
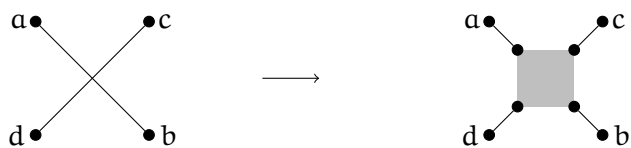

\subsection{Connected components in sparse graphs}
For the proof of our main result we need the fact that for some sparse graphs, induced subgraphs with sufficiently many edges are always connected. 
\begin{lemma}\label{lem:con}
Let $G = (V,E)$ be a $(k,l)$-sparse graph, where 
$1 \leq k \leq l \leq 2k - 1$.
Then any subset $S \subseteq V$ with
\begin{equation}
|E[S]| \geq k|S| - 2k + 1 
\end{equation}
induces a connected subgraph in $G$.
\end{lemma}
\begin{proof}
Assume $S$ would induce a subgraph in $G$ 
consisting of two disconnected components $C_1,C_2 \subseteq V$.
Since~$G$ is $(k,l)$-tight we have
\begin{align}
|E[S]| &= |E[C_1]| + |E[C_2]|\\
&\leq k|C_1| - l + k|C_2| - l\\
&= k|S| - 2l.
\end{align}
Since we assume that $k \leq l$ 
this contradicts the assumption that $|E[S]| \geq k|S| - 2k + 1$.
\end{proof}

\subsection{Replacing edges}
We show that for a $(k,l)$-tight graph, edges can be replaced by a (planar) graph of constant size, while preserving the $(k,l)$-tight property.

\begin{lemma}\label{lem:replaceEdge}
Let $G=(V,E)$ be a graph with an edge $(u,v) \in E$ and 
let $0 \leq l \leq 2k - 1$.
Let $\Omega = (V_{\Omega}, E_{\Omega})$ be a $(k,2k-1)$-tight graph with at least two vertices that we identify with~$u,v$.
Let $G^* = (V^*,E^*)$ be the graph obtained from~$G$
by removing edge~$(u,v)$ and taking the union of~$G$ and~$\Omega$
by identifying nodes~$u$ and~$v$ in ~$G$ and~$\Omega$.
That is, $G^* = G - (u,v) + \Omega$ for short.

Then $G$ is $(k,l)$-tight iff $G^*$ is $(k,l)$-tight.
\end{lemma}
\begin{proof}
Observe that 
\begin{align}\label{eq:nodeCount}
|V^*| &= |V| + |V_{\Omega}| -2,\\
|E^*| &= |E| + |E_{\Omega}| -1. \label{eq:edgeCount}
\end{align}
Assume $G$ is $(k,l)$-tight.
We have
\begin{align}
|E^*| &= k|V| - l - 1 + k|V_{\Omega}| - 2k + 1\\
&=k|V^*| - l.
\end{align}
Let $S \subseteq V^*$.
When $u,v \notin S$, set $S$ clearly has at most $k|S| - l$ many edges since~$G$ and~$\Omega$ are $(k,l)$-sparse.
Assume only one of $\{u,v\}$, say~$u$, is contained in~$S$.
Then we have
\begin{align}
|E^*[S]| &= |E[(S \setminus V_{\Omega}) \cup \{u\}]| + |E_{\Omega}[S\cap V_{\Omega}]|\\
&\leq k|(S \setminus V_{\Omega}) \cup \{u\}| - l + k|S\cap V_{\Omega}| - 2k + 1 \\
&\leq k|S| - l - k + 1\\
&\leq k|S| - l \label{eq:kgeq1}.
\end{align}
The last inequality follows since we assume that $k\geq 1$.

Now assume $u,v \in S$.
Then we have
\begin{align}
|E^*[S]| &= |E[(S \setminus V_{\Omega}) \cup \{u,v\}]| - 1 + |E_{\Omega}[S\cap V_{\Omega}]|\\
&\leq k|(S \setminus V_{\Omega}) \cup \{u,v\}| - l - 1 + k|S\cap V_{\Omega}| - 2k + 1 \\
&\leq k|S| - l.
\end{align}

For the other direction, assume that $G^*$ is $(k,l)$-tight.
By Equation~(\ref{eq:edgeCount}) we have 
\begin{align}
|E| &= |E^*| - |E_{\Omega}| + 1\\
&= (k|V^*| - l) - (k|V_{\Omega}| - 2k + 1) + 1\\
&=k|V| - l.
\end{align}

Let $S \subseteq V$.
Assume $u,v \in S$, otherwise $S$ would clearly have at most $k|S| - l$ edges since $G^*$ is $(k,l)$-tight.
Let $T = S \cup V_{\Omega}$.
Then we have
\begin{align}
|E[S]| 
&= |E^*[T]| + 1 - |E_{\Omega}| \\
&\leq (k|T| - l) + 1 - (k|V_{\Omega}| - 2k +1)\\
&= k|S| - l.
\end{align}
\end{proof}

When graph~$G$ has an edge crossing~$(a,b),~(c,d)$,
we put the gadget~$\Gamma$ to replace the crossing and obtain graph~$G'$.
By Lemma~\ref{lem:replaceEdge},
we may assume that the gadget~$\Gamma$ does not have
any of the edges within~$\{a,b,c,d\}$,
and the same will then hold for~$G'$.
Similarly it holds for~$G$,
except for~$(a,b),~(c,d)$ that are in~$G$.

\section{Proof of Theorem~\ref{thm:tightGadget}}\label{s:main}

We assume that $k \leq 3$, otherwise the statement of the Theorem is trivial since a $(k,l)$-tight graph would have too many edges to be planar and hence there is no planarizing gadget.
For the same reason, the cases $(3,l)$ are trivial for $l \leq 5$.
Therefore, it only remains to show that no planarizing gadgets exist for $k=2$ and $l=0,1,2,3$.
\ignore{
Given any edge in a graph we can replace the edge by a planar $(k,2k-1)$-tight graph as in Lemma~\ref{lem:replaceEdge} while preserving the $(k,l)$-tight property.
From this it follows that for an edge crossing \{ab,cd\} in a graph~$G$ we can assume without loss of generality that no edge from~$\binom{\{a,b,c,d\}}{2}$ is present in $G$ except for the crossing edges $ab,cd$.
In particular, there is no need to consider different gadgets depending on the the presence of different choices of edges from~$\binom{\{a,b,c,d\}}{2}$ in~$G$.
Moreover, we can also assume that there are no multi-edges present in $G$ since we can also remove them by Lemma~\ref{lem:replaceEdge}.}
For our proof we assume for contradiction that a planarizing gadget exists.
Clearly, any planarizing gadget must be $(k,l)$-sparse.
Let $G = (V,E)$ be graph with a crossing~$\{ab,cd\}$.
We denote $G' = G - \{ab, cd\} + \Gamma$ as the graph in which we have replaced the crossing with the planarizing gadget.
In Section~\ref{s:verySparse} we first show that certain vertex sets induce subgraphs in the gadget which are \emph{even more sparse}.
In contrast to that, we show the existence of some subgraphs in the gadget which have a lower bound on number of edges (Section~\ref{s:veryDense}).
Then we put everything together in Section~\ref{s:contradiction} by constructing a subgraph in the gadget where the upper bound is smaller than the lower bound from the previous sections.
Therefore we get a contradiction. 

To derive the above properties of the gadget, we use different choices for the graph $G$, since the gadget must work for \emph{all} graphs. 
For example, the gadget must also work for graphs where the end points of the crossing edges, $a,b,c,d$ are connected to each other even after removing the edges $ab,cd$ in $G$. Therefore in the planarizing gadget $\Gamma$, the vertices identified with $a,b,c,d$ must lie on a common face, 
otherwise~$G'$ would have edge crossings (it is useful to think of $\Gamma$ as lying inside a square bounding box with vertices $a,c,b,d$ in clockwise order). This implies in particular that inside $\Gamma$ there cannot be two vertex disjoint paths, one from $a$ to $b$, and another from $c$ to $d$, since $\Gamma$ must be planar. 

\subsection{Upper bounds on the number of edges in the gadget}\label{s:verySparse}

We first make an observation about the size of the gadget.
\begin{lemma}\label{lem:gadgetSize}
Assume that $\Gamma = (V_{\Gamma},E_{\Gamma})$ is a planarizing gadget for $(k,l)$-tight graphs.
Then we have
\begin{align}\label{eq:gadgetSize}
|E_{\Gamma}| = k|V_{\Gamma}|  - 4k + 2.
\end{align}
\end{lemma}

\begin{proof}
Let $G = (V,E)$ be a $(k,l)$-tight graph
and $G' = (V',E') = G - \{ab, cd\} + \Gamma$ be the graph obtained from~$G$ when putting~$\Gamma$.
Note that~$G'$ is also $(k,l)$-tight, 
otherwise the gadget would not be valid.
Hence we can compute
\begin{align}\label{eq:Gedges}
|E_{\Gamma}| & = |E'| - |E| + 2\\
&= k|V'| - l - k|V| + l + 2\\
&=k(|V_{\Gamma}| - 4) + 2\\
&=k|V_{\Gamma}|  - 4k + 2.
\end{align}
\end{proof}

Any subset of vertices~$S$ of $\Gamma$ can induce at most $k|S| - l$ many edges, otherwise the graph $G'$ would not be $(k,l)$-tight.
When we consider subsets $S \subseteq V_{\Gamma}$ that contain the vertices~$a,b,c,d$, we can strengthen the sparsity condition as follows.

\begin{lemma}\label{lem:sparse1}
Assume that $\Gamma = (V_{\Gamma},E_{\Gamma})$ is a planarizing gadget for $(k,l)$-tight graphs.
Then, for every subset $S\subseteq V_{\Gamma}$, where $a,b,c,d \in S$, we have
\begin{align}\label{eq:sparse1}
|\E_{\Gamma}[S]| \leq k|S| - 4k + 2.
\end{align}
\end{lemma}

\begin{proof}
Let $G = (V,E)$ be a $(k,l)$-tight graph
and $G' = (V',E') = G - \{ab, cd\} + \Gamma$ be the graph obtained from~$G$ when putting~$\Gamma$.
Let~$S \subseteq V_{\Gamma}$ be a set of vertices such that $a,b,c,d \in S$.
Since $G'= (V',E')$ must also be $(k,l)$-tight we have
\begin{align}
|E_{\Gamma}[S]|
&= |E'[S \cup V]|  - (|E| - 2) \\
&\leq k|S \cup V| - l - |E| + 2\\
&= k(|S| + |V| -4) - l - (k|V| - l) + 2\\
&= k|S| - 4k + 2.
\end{align}
\end{proof}

For $(2,0)$-tight graphs,
we need a similar statement when~$S$ contains just two
nodes from $\{a,b,c,d\}$,
namely one from $\{a,b\}$ and one from $\{c,d\}$.

\begin{lemma}\label{lem:sparse2}
Assume that $\Gamma=(V_{\Gamma},E_{\Gamma}) $ is a planarizing gadget for $(2,0)$-tight graphs.
Then, for every subset $S\subseteq V_{\Gamma}$, where
$|\{a,b\} \cap S| = 1$ and $|\{c,d\} \cap S| = 1$,
we have
\begin{align}\label{eq:sparse2}
|\E_{\Gamma}[S]| \leq 2|S| - 4.
\end{align}
\end{lemma}

\begin{proof}
Consider the $(2,0)$-tight graph~$G = (V,E)$ in Figure~\ref{fig:sparse2}.
Observe that $G \setminus \{b,c\}$ has~$10$ edges.
\begin{figure}[htbp]
\centering
    \usetikzlibrary{positioning, fit, shapes.geometric}

\tikzstyle{every node}=[circle, draw, fill=black,inner sep=0pt, minimum width=4pt,
node distance =1 cm and 1cm ] 
\begin{tikzpicture}[scale=1.5]

\def\dx{3.5}

\begin{scope}[shift={(\dx,0)}]

    \node (a) at (0,0) [label=above:$a$]{}; 
    \node (c) at (1,0) [label=above:$c$]{};
    \node (d) at (0,-1) [label=below:$d$]{};
    \node (b) at (1,-1) [label=below:$b$]{};

    \node (1) at (-1,0) []{};
    \node (2) at (-1,-1) []{}; 

    \node (3) at (-1.5,-0.5) []{};

    \draw[] (1) -- (2);
    \draw[] (1) -- (3);
    \draw[] (1) -- (a);
    \draw[] (1) -- (d);
    \draw[] (2) -- (a);
    \draw[] (2) -- (d);
    \draw[] (2) -- (3);
    \draw[] (3) -- (a);
    \draw[] (3) -- (d);
    \draw[] (a) -- (d);
    \draw[] (b) -- (a);
    \draw[] (b) -- (d);
    \draw[] (c) -- (a);
    \draw[] (c) -- (d);

     \node[below, fill=none, draw=none] at (0,-1) (a) {$(2,0)$-tight};
\end{scope} 

\begin{scope}[shift={(2*\dx,0)}]
    \node (a) at (0,0) [label=above:$a$]{}; 
    \node (c) at (1,0) [label=above:$c$]{};
    \node (d) at (0,-1) [label=below:$d$]{};
    \node (b) at (1,-1) [label=below:$b$]{};

    \node (1) at (-1,0) []{};
    \node (2) at (-1,-1) []{}; 

    \node (3) at (-1.5,-0.5) []{};

    \draw[] (1) -- (2);
    \draw[] (1) -- (3);
    \draw[] (1) -- (a);
    \draw[] (1) -- (d);
    \draw[] (2) -- (a);
    \draw[] (2) -- (d);
    \draw[] (2) -- (3);
    \draw[] (3) -- (a);
    \draw[] (3) -- (d);
    \draw[] (b) -- (a);
    \draw[] (b) -- (d);
    \draw[] (c) -- (a);
    \draw[] (c) -- (d);

      \node[below, fill=none, draw=none] at (0,-1) (b) {$(2,1)$-tight};
\end{scope}
 \end{tikzpicture}
        \caption{The graphs used in the proofs of Lemma~\ref{lem:sparse2} and Lemma~\ref{lem:2}}.
        \label{fig:sparse2}
\end{figure}
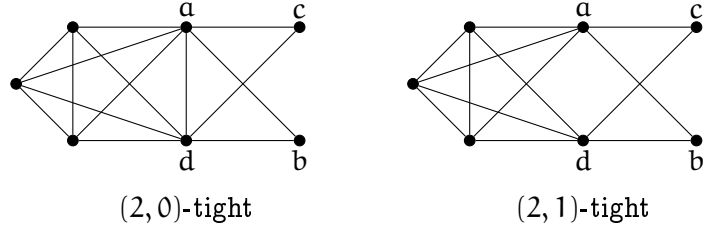
Let $S \subseteq V_{\Gamma}$, where $a,d \in S$ and $b,c \notin S$.
We have
\begin{align}
|E_{\Gamma}[S]|
&=  |E'[S \cup V \setminus \{b,c\}]|  -  10\\
&\leq 2|S \cup V \setminus \{b,c\}| - 10\\
&= 2(|S| + 3) - 10\\
&=2|S| - 4.
\end{align}
By relabeling the vertices~$a,b,c,d$ in the graph from Figure~\ref{fig:sparse2} we obtain the bound for all cases from the statement.
\end{proof}

\subsection{Lower bounds on the number of edges in the gadget}\label{s:veryDense}
When the given graph~$G$ is not $(k,l)$-tight,
graph~$G'$ should also be not $(k,l)$-tight. 
This means that $G'$ must contain a subset of vertices that violates the $(k,l)$-sparsity condition. 
The following lemma is useful for such an argument.

\begin{lemma}\label{lem:2}
Assume that $\Gamma=(V_{\Gamma},E_{\Gamma})$ is a planarizing gadget for $(2,l)$-tight graphs where $l = 0,1,2,3$.
Then there are two subsets $S_1,S_2\subseteq V_{\Gamma}$ such that
\begin{align}
a,b \in S_1,~ c,d \notin S_1 &~\text{ and }~ |E_{\Gamma}[S_1]| \geq 2|S_1| - 3,\\
c,d \in S_2,~ a,b \notin S_2 &~\text{ and }~ |E_{\Gamma}[S_2]| \geq 2|S_2| - 3.
\end{align}
Moreover, the subgraphs $\Gamma[S_1],\Gamma[S_2]$ contain a path from $a$ to $b$ and $c$ to $d$, respectively.
\end{lemma}
\begin{proof}
We first show the existence of the set~$S_1$.
The existence of~$S_2$  follows by a similar argument when we change the labeling of the vertices in the crossing by mapping $a,b,c,d$ to $c,d,b,a$, respectively.
Consider Figure~\ref{fig:Gh1h2} where we have three graphs $G,H_1,H_2$.
\begin{figure}[ht]
\begin{center}
    \usetikzlibrary{positioning, fit, shapes.geometric}

\tikzstyle{every node}=[circle, draw, fill=black,inner sep=0pt, minimum width=4pt,
node distance =1 cm and 1cm ] 
\begin{tikzpicture}[scale=1.2]

\def\dx{3}

\begin{scope}[shift={(\dx,0)}]
    \node (a) at (0,0) [label=left:$a$]{}; 
    \node (c) at (0,-2) [label=left:$c$]{};
    \node (d) at (1,-2) [label=right:$d$]{};
    \node (b) at (1,0) [label=right:$b$]{};

    \node (e) at (0,-1) [label=left:$e$]{};
    \node (f) at (1,-1) [label=right:$f$]{};

	\draw[] (a) -- (b);
	\draw[] (c) -- (d);

    \draw[] (d) -- (f);
    \draw[] (e) -- (f);
    \draw[] (a) -- (f);
    \draw[] (a) -- (e);
    \draw[] (b) -- (f);
    \draw[] (b) -- (e);
    \draw[] (c) -- (e);

       \node[below, fill=none, draw=none] at (0.5,-2.6) (a) {$G$};

\end{scope}

\begin{scope}[shift={(2*\dx,0)}]
    
    \node (a) at (0,0) [label=left:$a$]{}; 
    \node (c) at (0,-2) [label=left:$c$]{};
    \node (d) at (1,-2) [label=right:$d$]{};
    \node (b) at (1,0) [label=right:$b$]{};

    \node (e) at (0,-1) [label=left:$e$]{};
    \node (f) at (1,-1) [label=right:$f$]{};

	\draw[] (a) -- (b);
	\draw[] (c) -- (d);

    \draw[] (d) -- (f);
    \draw[] (e) -- (f);
    \draw[] (a) -- (f);
    \draw[] (d) -- (e);
    \draw[] (b) -- (f);
    \draw[] (b) -- (e);
    \draw[] (c) -- (e);

       \node[below, fill=none, draw=none] at (0.5,-2.6) (b) {$H_1$};
   
\end{scope}

\begin{scope}[shift={(3*\dx,0)}]
    
    \node (a) at (0,0) [label=left:$a$]{}; 
    \node (c) at (0,-2) [label=left:$c$]{};
    \node (d) at (1,-2) [label=right:$d$]{};
    \node (b) at (1,0) [label=right:$b$]{};

    \node (e) at (0,-1) [label=left:$e$]{};
    \node (f) at (1,-1) [label=right:$f$]{};

	\draw[] (a) -- (b);
	\draw[] (c) -- (d);

    \draw[] (d) -- (f);
    \draw[] (e) -- (f);
    \draw[] (a) -- (f);
    \draw[] (a) -- (e);
    \draw[] (c) -- (f);
    \draw[] (b) -- (e);
    \draw[] (c) -- (e);

       \node[below, fill=none, draw=none] at (0.5,-2.6) (c) {$H_2$};
\end{scope}

 \end{tikzpicture}
  
        \caption{The graphs for Lemma~\ref{lem:2} for $l=3$. Graph $G$ is not $(2,3)$-tight, but $H_1, H_2$ are. 
        }
        \label{fig:Gh1h2}
\end{center}  
\end{figure}
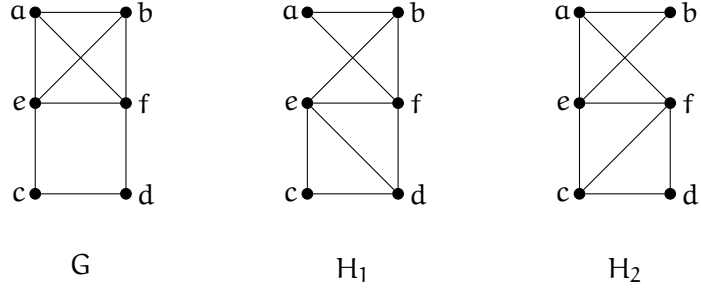
The graph $G = (V,E)$ is a not $(2,3)$-tight because the set $V - \{c,d\}$ violates the sparsity condition.
When we move the edge $(a,e)$ to $(d,e)$ we obtain a new graph~$H_1$ which is $(2,3)$-tight.
Analogously, 
we obtain a graph~$H_2$  by moving the edge $(b,f)$ to $(c,f)$ which is also $(2,3)$-tight.

We can construct similar graphs for the cases $l=0,1,2$, see Figure~\ref{fig:otherGs}.

\begin{figure}[htb]
\begin{center}
    \usetikzlibrary{positioning, fit, shapes.geometric}

\tikzstyle{every node}=[circle, draw, fill=black,inner sep=0pt, minimum width=4pt,
node distance =1 cm and 1cm ] 
\begin{tikzpicture}[scale=1.2]

\def\dx{3}

\begin{scope}[shift={(\dx,0)}]
   \node (a) at (0,0) [label=left:$a$]{}; 
    \node (c) at (0,-2) [label=left:$c$]{};
    \node (d) at (1,-2) [label=right:$d$]{};
    \node (b) at (1,0) [label=right:$b$]{};

    \node (e) at (0,-1) [label=left:$e$]{};
    \node (f) at (1,-1) [label=right:$f$]{};

    \node (x) at (0.25,0.5) []{};
    \node (y) at (0.75,0.5) []{};

    \draw[] (a) -- (x);
    \draw[] (a) -- (y);
    \draw[] (e) -- (y);
    \draw[] (b) -- (y);
    \draw[] (b) -- (x);
    \draw[] (f) -- (x);
    \draw[] (x) -- (y);

	\draw[] (a) -- (b);
	\draw[] (c) -- (d);

    \draw[] (d) -- (f);
    \draw[] (e) -- (f);
    \draw[] (a) -- (f);
    \draw[] (a) -- (e);
    \draw[] (b) -- (f);
    \draw[] (b) -- (e);
    \draw[] (c) -- (e);

       \node[below, fill=none, draw=none] at (0.5,-2.6) (a) {$l=0$};

\end{scope}

\begin{scope}[shift={(2*\dx,0)}]
   \node (a) at (0,0) [label=left:$a$]{}; 
    \node (c) at (0,-2) [label=left:$c$]{};
    \node (d) at (1,-2) [label=right:$d$]{};
    \node (b) at (1,0) [label=right:$b$]{};

    \node (e) at (0,-1) [label=left:$e$]{};
    \node (f) at (1,-1) [label=right:$f$]{};

        \node (x) at (0.5,0.5) []{};

    \draw[] (a) -- (x);
    \draw[] (b) -- (x);
    \draw[] (e) -- (x);
    \draw[] (f) -- (x);

	\draw[] (a) -- (b);
	\draw[] (c) -- (d);

    \draw[] (d) -- (f);
    \draw[] (e) -- (f);
    \draw[] (a) -- (f);
    \draw[] (a) -- (e);
    \draw[] (b) -- (f);
    \draw[] (b) -- (e);
    \draw[] (c) -- (e);

       \node[below, fill=none, draw=none] at (0.5,-2.6) (a) {$l=1$};
   
\end{scope}

\begin{scope}[shift={(3*\dx,0)}]
    
    \node (a) at (0,0) [label=left:$a$]{}; 
    \node (c) at (0,-2) [label=left:$c$]{};
    \node (d) at (1,-2) [label=right:$d$]{};
    \node (b) at (1,0) [label=right:$b$]{};

    \node (e) at (0,-1) [label=left:$e$]{};
    \node (f) at (1,-1) [label=right:$f$]{};

    \node (x) at (0.5,0.5) []{};

    \draw[] (a) -- (x);
    \draw[] (b) -- (x);
    \draw[] (e) -- (x);
    \draw[] (f) -- (x);

	\draw[] (a) -- (b);
	\draw[] (c) -- (d);

    \draw[] (d) -- (f);
    \draw[] (a) -- (f);
    \draw[] (a) -- (e);
    \draw[] (b) -- (f);
    \draw[] (b) -- (e);
    \draw[] (c) -- (e);

       \node[below, fill=none, draw=none] at (0.5,-2.6) (a) {$l=2$};
\end{scope}

 \end{tikzpicture} 
        \caption{The graphs $G$ for Lemma~\ref{lem:2} for  $l=0,1,2$.
        The corresponding graphs $H_1$ and $H_2$ are defined analogously as in the case $l=3$ in Figure~\ref{fig:Gh1h2}.
        }
        \label{fig:otherGs}
\end{center}   
\end{figure}
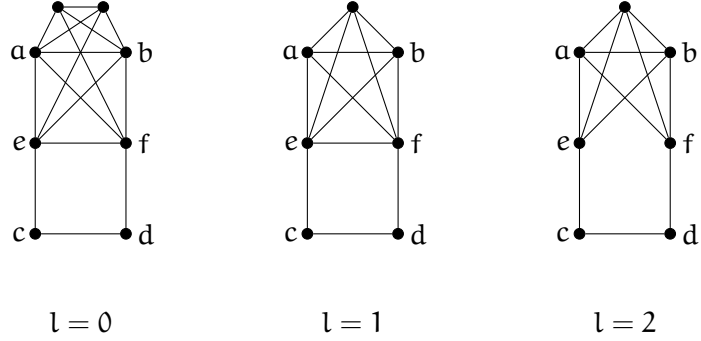
Since $G$ is not $(2,l)$-tight, the graph~$G' = (V', E')$ 
where we have put the gadget, must also be not $(2,l)$-tight.
Hence, there exists a set $T \subseteq V'$ that violates the $(2,l)$-sparsity condition,
\begin{equation}\label{eq:T}
|E'[T]| > 2|T| -l.
\end{equation}
However, the graphs~$H_1'$ and~$H_2'$ where we put the gadget must both be $(2,l)$-tight and thus, 
the same set~$T$ must fulfill the $(2,l)$-sparsity condition in~$H_1',H_2'$.
That is, 
$|E_{H_i'}[T]| \leq 2|T| -l$, for $i = 1,2$.
The difference  between~$G$ and $H_1,H_2$ is 
that one edge has moved.
It follows that $a,b \in T$ and $c,d \notin T$, otherwise the edge moves would not decrease the number of edges in~$E_{H_1'}[T]$ and~$ E_{H_2'}[T]$.

\ignore{
We also get that~$T$ violates the sparsity condition in~$G'$ by~$1$,
\begin{equation}\label{eq:T}
|E'[T]| = 2|T| -l+1.
\end{equation}
}

Let $S_1 = T \cap V_{\Gamma}$ be the nodes of~$T$ inside the gadget.
Recall that we defined~$\Gamma$ such that $a,b,c,d \in V_{\Gamma}$.
Hence, we have $a,b \in S_1$ and $c,d \not\in S_1$, as required by the lemma.
Now we want to argue the lower bound on the number of edges induced by~$S_1$.

Let $T_1 = T \setminus S_1 \cup \{a,b\}$
be the nodes of~$T$ outside the gadget, but including~$a,b$.
By Lemma~\ref{lem:replaceEdge},
we may assume that $(a,b) \not\in E_{\Gamma}$.
Note that~$G$ becomes $(2,l)$-sparse when we remove the edge~$(a,b)$.
Hence, $G'[T_1]$ must be $(2,l)$-sparse,
\begin{equation}\label{eq:T_1}
|E'[T_1]| \leq 2 |T_1| -l.
\end{equation}

By definition, we have  $T = T_1 \cup S_1$.
Since $(a,b) \not\in E'$,
from~(\ref{eq:T}) and~(\ref{eq:T_1}), we get
\begin{align}
|E_{\Gamma}[S_1]| = |E'[S_1]| &= |E'[T]| - |E'[T_1]| \\
&\geq  (2|T| -l+1) - (2|T_1| -l)\\
&=  2|T| +1 - 2|T_1| \\
&=  2|T| +1 - 2( |T| - |S_1| + 2)\\
&= 2|S_1| -3.
\end{align}

It remains to argue the last part that the subgraph $\Gamma[S_1]$ contains a path from~$a$ to~$b$.
By Lemma~\ref{lem:con},
the subgraph $\Gamma[S_1]$ is connected  for  $l \geq 2$
and we are done. 

For $l=0,1$, 
let $G$ be the $(2,0)$-tight or the $(2,1)$-tight graph in Figure~\ref{fig:sparse2}.
Assume there is no path from~$a$ to~$b$ in~$S_1$.
Hence,
we can partition~$S_1$ into two components~$A,B$
such that $a \in A$ and $b \in B$,
and $\Gamma[A]$ and $\Gamma[B]$ are disconnected.
Since $G$ is $(2,l)$-tight, $G'$ must also be $(2,l)$-tight.

Choose $X = A \cup V - \{b,c\}$. 
For both graphs~$G$ in Figure~\ref{fig:sparse2}, for  $l=0,1$,
we have that $|X| = |A| + |V- \{b,c\}| -1 = |A| + 4$,
because~$A$ and~$V$ have one node in common,~$a$.
Also,
without nodes~$b,c$,
the graphs have $|E[V-\{b,c\}]| = 10 - l$ edges.
Since~$X$ must fulfill the sparsity condition in~$G'$,
we get
\begin{align}
|E_{\Gamma}[A]|  &= |E'[X]| - |E[V-\{b,c\}]| \\
&=  |E'[X]| - (10 -l)\\
&\leq 2|X| - l - 10 + l\\
&= 2(|A| +4) - 10\\
&= 2|A| -2.
\end{align}

When we change the labeling of the vertices in the example by mapping $a$ to $b$ we get the same bound for $|E_{\Gamma}[B]|$ as well.
But then we get a contradiction:
\begin{align}
2|S_1| - 3 &\leq |E_{\Gamma}[S_1]|\\
&= |E_{\Gamma}[A]| + |E_{\Gamma}[B]|\\
&\leq 2|A| - 2 + 2|B| - 2 \\
&=2|S_1| - 4.
\end{align}
\end{proof}

\subsection{
Getting a contradiction from upper and lower bound
}\label{s:contradiction}
In this section, we prove Theorem~\ref{thm:tightGadget} by showing that the upper bound from Lemma~\ref{lem:sparse1} contradicts the lower bound from Lemma~\ref{lem:2}:
Consider the two subsets $S_1,S_2$ from Lemma~\ref{lem:2}.
Note that $a,b,c,d$ are contained in $S_1 \cup S_2$.
We show that $S_1 \cup S_2$ violates Lemma~\ref{lem:sparse1}.
We distinguish different cases depending on the cardinality of $S_1 \cap S_2$.

\paragraph{Case 1.} $S_1 \cap S_2 = \emptyset$.
By Lemma~\ref{lem:2} we know that the subgraphs induced by $S_1, S_2$ contain a path $P_{ab}$ from $a$ to $b$ and a path $P_{cd}$ from $c$ to $d$, respectively.
A planar embedding of the gadget graph~$\Gamma$ where the vertices~$a,c,d,b$ appear on the outerface in that order looks like the one in Figure~\ref{fig:gadgetDrawing}.
Since we assume $S_1 \cap S_2 = \emptyset$, the sets $S_1,S_2$ have no common vertices and thus, the paths $P_{ab}, P_{cd}$ must cross with each other.
Therefore, the case~$S_1 \cap S_2 = \emptyset$ does not emerge as the gadget would not be planarizing.
\begin{figure}[htbp]
\centering
    \usetikzlibrary{positioning, fit, shapes.geometric}
\usetikzlibrary{calc}

\tikzstyle{every node}=[circle, draw, fill=black,inner sep=0pt, minimum width=4pt,
node distance =1 cm and 1cm ] 
\begin{tikzpicture}[scale=3.2]
    \node (a) at (0,0) [label=left:$a$]{}; 
    \node (c) at (1,0) [label=right:$c$]{};
    \node (d) at (0,-1) [label=left:$d$]{};
    \node (b) at (1,-1) [label=right:$b$]{};

    \begin{scope}[shift={(a))},x={(b)},y={($(a)!1!90:(b)$)}]
    \draw[] (.5,0) ellipse (.5 and .25);
    \end{scope}

    \begin{scope}[shift={(d))},x={(c)},y={($(d)!1!90:(c)$)}]
    \draw[] (.5,0) ellipse (.5 and .25);
    \end{scope}

    \node[below, fill=none, draw=none] at (0.75,-1.2) (s1) {$S_1$};
    \node[below, fill=none, draw=none] at (0.25,-1.2) (s2) {$S_2$};

    \node[below, fill=none, draw=none] at (0.15,-0.05) (sa) {$S_{a}$};
    \node[below, fill=none, draw=none] at (0.85,-0.75) (sb) {$S_{b}$};
    \node[below, fill=none, draw=none] at (0.85,-0.05) (sc) {$S_{c}$};
    \node[below, fill=none, draw=none] at (0.15,-0.75) (sd) {$S_{d}$};
    \node[below, fill=none, draw=none] at (0.5,-0.3) (sd) {$S_1 \cap S_2$};

 \end{tikzpicture}
        \caption{Drawing of the gadget graph~$\Gamma$ with the sets $S_1,S_2$ from Lemma~\ref{lem:2}.}
        \label{fig:gadgetDrawing}
\end{figure}
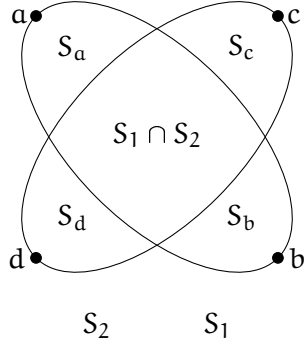

 \paragraph{Case 2.} $|S_1 \cap S_2| = 1$.
 Then we have
 \begin{align}
 |E_{\Gamma}[S_1 \cup S_2]| &\geq |E_{\Gamma}[S_1]| + |E_{\Gamma}[S_2]|\label{est:case21} \\
 & \geq 2|S_1| - 3 + 2|S_2| - 3\label{est:case22}\\
 & = 2(|S_1 \cup S_2| + 1) - 6\\
 & = 2|S_1 \cup S_2| - 4
 \end{align}
 Inequality~(\ref{est:case22}) follows from Lemma~\ref{lem:2} and contradicts Lemma~\ref{lem:sparse1}.

\paragraph{Case 3.} $|S_1 \cap S_2| \geq 2$.
Then we have

\begin{align}
|E_{\Gamma}[S_1 \cup S_2]| & \geq |E_{\Gamma}[S_1]| + |E_{\Gamma}[S_2]| - |E_{\Gamma}[S_1 \cap S_2]|\label{estimate1}\\
& \geq 2|S_1| - 3 + 2|S_2| - 3 - (2|S_1 \cap S_2| - l)\label{estimate2}\\
& = 2|S_1 \cup S_2| - 6 + l.\label{estimate3}
\end{align}
Inequality~(\ref{estimate2}) follows from Lemma~\ref{lem:2} and the fact that $\Gamma[S_1 \cap S_2]$ must be $(k,l)$-sparse.
Inequality~(\ref{estimate3}) contradicts Lemma~\ref{lem:sparse1} when $l \geq 1$.

The case $l=0$ is more elaborate.
Since~$\Gamma$ is a planarizing gadget it must admit a planar embedding where the vertices $a,c,b,d$ appear on the outer-face in that order, as in Figure~\ref{fig:gadgetDrawing}.
Observe that in this embedding, the set $S_2$ naturally partitions the set $S_1$ into three parts $S_{a}, S_1 \cap S_2, S_{b}$ such that $a \in S_{a}, b \in S_{b}$.
Note that there cannot be any edges in the gadget with one end point in $S_{a}$ and other end point in $S_{b}$.
Vice versa, $S_1$ partitions $S_2$ into $S_{c}, S_1 \cap S_2, S_{d}$.
We give an upper bound on the number of edges in $E_{\Gamma}[S_1]$ that contradicts Lemma~\ref{lem:2}.

For $x \in \{a,b,c,d\}$ we define
\[
S_{x}' = S_{x} \cup (S_1 \cap S_2).
\]
and we let $z_x \geq 0$ be an integer such that
\begin{align}
|E_{\Gamma}[S_{x}']| = 2|S_{x}'| - z_x.
\end{align}
Note that we can assume that there are no edges going from $S_{a}$ to $S_{c}, S_{d}$ or from $S_{b}$ to $S_{c}, S_{d}$, 
otherwise~(\ref{estimate1}) would be larger and thus we get a contradiction.
For the same reason, we can assume $|E_{\Gamma}[S_1 \cap S_2]| = 2|S_1 \cap S_2|$ by~(\ref{estimate2}).
Thus, for the two integers $z_a, z_d \geq 0$, we have
\begin{align}
|E_{\Gamma}[S_{a}' \cup S_{d}']| &= |E_{\Gamma}[S_{a}']| + |E_{\Gamma}[S_{d}']| - |E_{\Gamma}[S_1 \cap S_2]| \\
&= |E_{\Gamma}[S_{a}']| + |E_{\Gamma}[S_{d}']| - 2|S_1 \cap S_2|\\
&= 2|S_{a}'| - z_a + 2|S_{d}'| - z_d - 2|S_1 \cap S_2|\\
&= 2|S_{a}' \cup S_{d}'| - (z_a + z_d).
\end{align}
By Lemma~\ref{lem:sparse2}, 
we also have $|E_{\Gamma}[S_{a}' \cup S_{d}']| \leq 2|S_{a}' \cup S_{d}'| - 4$. 
Thus, we get
\begin{align}
z_a + z_d \geq 4.
\end{align}
When we apply the same arguments to the pairs $(S_{a}',S_{c}'),~ (S_{b}',S_{d}'),~(S_{b}',S_{c}')$ we obtain the constraints
\begin{align}
z_a + z_c &\geq 4,\\
z_b + z_d &\geq 4,\\
z_b + z_c &\geq 4.
\end{align}
It follows that $z_a + z_b \geq 4$ or $z_c + z_d \geq 4$.
Assume $z_a + z_b \geq 4$.
Since there are no edges between $S_a,S_b$, we have 
\begin{align}
 |E_{\Gamma}[S_1]| &= |E_{\Gamma}[S_{a}']| + |E_{\Gamma}[S_{b}']| - |E_{\Gamma}[S_1 \cap S_2]|\\
&= 2|S_{a}'| - z_a + 2|S_{b}'| - z_b - 2|S_1 \cap S_2|\\
&=2|S_1| - (z_a + z_b)\\
&\leq 2|S_1| - 4
\end{align}

But from Lemma~\ref{lem:2}, we must have $|E_{\Gamma}[S_1]| \geq 2|S_1| - 3$, which gives a contradiction.
This finishes the proof of Theorem~\ref{thm:tightGadget}.

\section{Discussion}\label{s:discussion}

Theorem~\ref{thm:tightGadget} also holds for sparse graphs instead of tight graphs.
Lemma~\ref{lem:replaceEdge} can be slightly modified
so that edges can be replaced with a constant size graph while preserving sparsity.
When we modify the proof of Theorem~\ref{thm:tightGadget} for sparsity, the equality in Lemma~\ref{lem:gadgetSize} becomes an upper bound and the other arguments still work. 

\begin{theorem}\label{thm:sparseGadget}
Planarizing gadgets do not exist for
$(k,l)$-sparse graphs,
for $k \geq 2$ and $0 \leq l \leq 2k - 1$.
\end{theorem}
Note that Theorem~\ref{thm:sparseGadget} is not a generalization of Theorem~\ref{thm:tightGadget}, the two statements are incomparable.
This is because for a $(k,l)$-tight graph, the graph $G'$ where we replaced a crossing with a gadget can become $(k,l)$-sparse in terms of Theorem~\ref{thm:sparseGadget}.

Our arguments do not work directly when we replace \emph{tight} with \emph{spanning} in Theorem~\ref{thm:tightGadget} since the size of the gadget can become arbitrarily large in this case. However, note that the situation is different in this case since it is not known whether we can recognize $(k,l)$-spanning graphs in~$\NC$ in the planar case.

Our result rules out a common approach to reduce the problem of deciding whether a given graph is $(k,l)$-tight or sparse to the planar case. However, this does not imply that such a reduction does not exist in general.
For example, we assume that the gadget should work for all possible drawings of a given graph.
One might still be able to compute a drawing of the graph in $\NC$ or $\Logspace$ such that a planarizing gadget is possible.
This situation arises analogously for bipartite perfect matching and Hamiltonian cycle~\cite{matchingGadget}. 
An example is illustrated in Figure~\ref{fig:diffDrawing}.

\begin{figure}[bth]
\centering
    \usetikzlibrary{positioning, fit, shapes.geometric}

\tikzstyle{every node}=[circle, draw, fill=black,inner sep=0pt, minimum width=4pt,
node distance =1 cm and 1cm ] 
\begin{tikzpicture}[scale=1]

\def\dx{4}

\begin{scope}[shift={(\dx,0)}]
    \node (1) at (0,0) []{}; 
    \node (2) at (1,1) [label=above left:$a$]{};
    \node (3) at (2,1) []{};
    \node (4) at (3,0) []{};
    \node (5) at (2,-1) [label=below right:$b$]{};
    \node (6) at (1,-1) []{};

      \node (7) at (1.25,-0.5) [label=below right:$d$]{}; 
    \node (8) at (1.75,0.5) [label=below right:$c$]{}; 

    \draw (1) -- (2);
    \draw (2) -- (3);
    \draw (3) -- (4);
    \draw (4) -- (5);
    \draw (5) -- (6);
    \draw (6) -- (1);
    \draw (2) -- (5);
    \draw (6) -- (3);

    \draw (1) to [out=90,in=90,looseness=1.8] (4);
    \draw (1) to [out=80,in=130,looseness=1.3] (3);
    \draw (1) to [out=280,in=230,looseness=1.3] (5);

    \node[below, fill=none, draw=none] at (1.5,-1.7) (a) {(a)};
\end{scope}

\begin{scope}[shift={(2*\dx,0)}]
    \node (1) at (0,0) []{}; 
    \node (2) at (1,1) [label=above left:$a$]{};
    \node (3) at (2,1) [label=above right:$c$]{};
    \node (4) at (3,0) []{};
    \node (5) at (2,-1) [label=below right:$b$]{};
    \node (6) at (1,-1) []{};

      \node (7) at (1.1,-0.8) []{}; 
    \node (8) at (1.35,-0.3) [label=below right:$d$]{}; 

    \draw (1) -- (2);
    \draw (2) -- (3);
    \draw (3) -- (4);
    \draw (4) -- (5);
    \draw (5) -- (6);
    \draw (6) -- (1);
    \draw (2) -- (5);
    \draw (6) -- (3);

    \draw (1) to [out=90,in=90,looseness=1.8] (4);
    \draw (1) to [out=80,in=130,looseness=1.3] (3);
    \draw (1) to [out=280,in=230,looseness=1.3] (5);

    \node[below, fill=none, draw=none] at (1.5,-1.7) (b) {(b)};
\end{scope}

 \end{tikzpicture}
        \caption{The same non-$(2,3)$-tight graph in two different drawings: (a) Adding a gadget by deleting the edges $ab, cd$ and connecting a new vertex $x$ to $a,b,c,d$ makes the graph $(2,3)$-tight. (b) the same gadget preserves the non-$(2,3)$-tightness.}
        \label{fig:diffDrawing}
\end{figure}
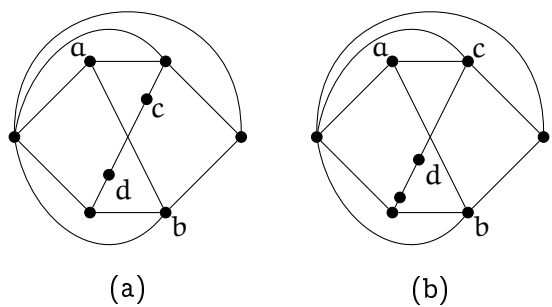

\section*{Acknowledgments}
We thank Meera Sitharam for pointing out the reduction from $(k,l)$-sparsity to max flow.

\bibliographystyle{alpha}
\bibliography{references}

\clearpage
\appendix
\appendix

\section{Reducing $(k,l)$-sparsity to max flow with multiple sources and sinks}\label{app:klsparse_to_flow}

We describe a reduction from $(k,l)$-sparsity to \emph{max flow with multiple sources and sinks (MSMS max flow)}.
There are similar reductions in the literature, see e.g.~\cite{HLS97, HLS98}.
Still the precise setting here needs some adaptions.

Max flow with multiple sources and sinks is defined as follows.
\begin{definition}[MSMS max flow]
  We are given a tuple $(G,S,T,c)$, where
  \begin{itemize}
   \item $G=(V,E)$ is a directed graph,
   \item $S,T \subseteq V$ are disjoint sets of
     \emph{sources} and \emph{sinks}, respectively.
     The vertices in $S$ have no incoming edges and the vertices in $T$ have no outgoing edges in $G$,
   \item $c:E\rightarrow \mathbb{Z}_{\geq 0}$ is the \emph{capacity function}.
  \end{itemize}
  A \emph{flow} in the graph is a function $f:E\rightarrow \mathbb{R}_{\geq 0}$. 
  For $v \in V$, the \emph{net outgoing flow} from~$v$ is
  denoted by~$f_v$,
  \[
  f_v = \sum\limits_{(v,w)\in E}f(v,w) -
  \sum\limits_{(w,v)\in E}f(w,v).
  \]
  A flow~$f$ is \emph{feasible}, if
  \begin{itemize}
    \item $\forall \, e\in E~~ f(e) \leq c(e)$,
    \item $\forall \, v \in V - (S\cup T) ~~ f_v = 0$.
  \end{itemize}
  Given a feasible flow~$f$, its \emph{value} is defined as
  $|f| = \sum_{s \in S} f_s$.
The objective of the problem is
  to find the maximum value of a feasable flow for~$(G,S,T,c)$. 
\end{definition}

We show  a reduction from sparseness to MSMS flows.

\begin{lemma}\label{claim:klSparse_to_flow}
Let $G = (V,E)$ be graph with $|V|=n$ and $|E|= m$,
and $0 \leq l \leq 2k - 1$. Deciding whether  $G$ is $(k,l)$-sparse reduces to solving $m$ many MSMS flow problems.

Moreover,
when $G$ is planar, has bounded genus or has bounded treewidth, these flow problems will have the corresponding property too.
\end{lemma}

\begin{proof}
We construct a MSMS flow problem $(H,S,T,c)$ as follows.
Graph $H = (V_H, E_H)$ has vertices
$V_H = S \cup E\cup V\cup T$,
where~$S$ is a copy of~$E$ and~$T$ is a copy of~$V$.
That is,
for $E = \{e_1,e_2,\dots, e_m\}$,
let $S = \{s_1,s_2,\dots, s_m\}$,
where $s_i \in S$ is associated with $e_i \in E$.
Similarly,
for $V = \{v_1,v_2,\dots, v_n\}$,
let $T = \{t_1,t_2,\dots, t_n\}$,
where $t_i \in T$ is associated with $v_i \in V$.

The edges of~$H$ are
\[
E_H = E_1 \cup E_2 \cup E_3,
\]
where
\begin{itemize}
\item $E_1$ are the $S$-$E$ edges,
$E_1 = \set{(s_i,e_i)}{i \in [m]}$
with capacities
$c(s_i,e_i) = 1$.
\item 
$E_2$ are the $E$-$V$ edges,
$E_2 = \set{(e,v)}{e \in E, ~v \in e}$
with capacities $c(e,v) = \alpha = m + l + 1$. 
The value of~$\alpha$ is chosen large enough
such that the $E_2$-edges can never appear in an $S$-$T$-cut of value at most~$m + l$.
\item 
$E_3$ are the $V$-$T$ edges,
$E_3 = \set{(v_i,t_i)}{i \in [n]}$
with capacities $c(v_i,t_i) = k$.
\end{itemize}

Observe that~$H$ is constructed from~$G$ by subdividing edges or adding vertices of degree one. Hence, the construction preserves the genus,
in particular planarity, and the treewidth of the input graph~$G$.

Flow problem $(H,S,T,c)$ serves as the base
from which we define~$m$ variants that we actually use in the reduction.
For each edge $e \in E$, 
we define essentially the same flow problem $(H,S,T,c_e)$,
only the capacity function is changed for one edge:
Let $s \in S$ be the copy of $e \in E$.
Then we define $c_e$ to be same as~$c$ except for edge~$(s,e)$ where we set
\[
c_e (s,e) = l+1.
\]
\ignore{
\[
c_{\hat{e}}(s_i,e_i) = 
\begin{cases}
 l+1,& \text{if } e_i = \hat{e},\\
 1,& \text{otherwise}.
\end{cases}
\]
}
The construction of $(H,S,T,c_e)$ is shown in Figure~\ref{fig:klsparse_to_flow}.
\begin{figure}[ht]
  \centering
  \usetikzlibrary{positioning, fit, shapes.geometric, arrows.meta, arrows.spaced}

\usetikzlibrary{decorations.markings}

\tikzstyle{every node}=[circle, draw, fill=black,inner sep=0pt, minimum width=4pt,
node distance =1 cm and 1cm ] 
\begin{tikzpicture}[scale=1, decoration={
    markings,
    mark=at position 0.5 with {\arrow{Latex[length=2mm,width=1.5mm]}}}]

    \node (s1) at (-1,0) [label=below left:$s_1$]{}; 
    \node (s2) at (-1,-1) [label=below left:$s_2$]{};
    \node (si) at (-1,-3) [label=below left:$s_i$]{};
    \node (sm) at (-1,-5) [label=below left:$s_m$]{};
    
    \node (e1) at (0,0) [label=below left:$e_1$]{}; 
    \node (e2) at (0,-1) [label=below left:$e_2$]{};
    \node (ei) at (0,-3) [label=below left:$e_i$]{};
    \node (em) at (0,-5) [label=below left:$e_m$]{};

    \draw[postaction={decorate}] (s1) -- node[fill=none, draw=none,above,sloped] {$1$} (e1);
    \draw[postaction={decorate}] (s2) -- node[fill=none, draw=none,above,sloped] {$1$} (e2);
    \draw[postaction={decorate}] (si) -- node[fill=none, draw=none,above,sloped] {$1+l$} (ei);
    \draw[postaction={decorate}] (sm) -- node[fill=none, draw=none,above,sloped] {$1$} (em);

    \draw[very thick,dotted] (0,-1.75) -- (0,-2.25);
    \draw[very thick,dotted] (0,-3.75) -- (0,-4.25);

    \node (v1) at (3,0) [label=below right:$v_1$]{}; 
    \node (v2) at (3,-2) [label=below right:$v_2$]{};
    \node (vn) at (3,-5) [label=below right:$v_n$]{};

    \node (t1) at (4,0) [label=below right:$t_1$]{}; 
    \node (t2) at (4,-2) [label=below right:$t_2$]{};
    \node (tn) at (4,-5) [label=below right:$t_n$]{};

    \draw[postaction={decorate}] (v1) -- node[fill=none, draw=none,above,sloped] {$k$} (t1);
    \draw[postaction={decorate}] (v2) -- node[fill=none, draw=none,above,sloped] {$k$} (t2);
    \draw[postaction={decorate}] (vn) -- node[fill=none, draw=none,above,sloped] {$k$} (tn);
    
    \draw[very thick,dotted] (3,-3) -- (3,-4);

    \draw[postaction={decorate}] (e1) -- node[fill=none, draw=none,above] {$\alpha$} (v1);
    \draw[postaction={decorate}] (e1) -- node[fill=none, draw=none,above] {~$\alpha$}(v2);
    \draw[postaction={decorate}] (e2) -- node[fill=none, draw=none,above] {$\alpha$}(v1);
    \draw[postaction={decorate}] (e2) -- node[fill=none, draw=none,right] {$\alpha$}(vn);
    \draw[postaction={decorate}] (em) -- node[fill=none, draw=none,above] {$\alpha$}(vn);
    \draw[postaction={decorate}] (ei) -- node[fill=none, draw=none,above] {$\alpha$}(v2);

     \draw[postaction={decorate}] (ei) -- node[fill=none, draw=none,above] {~$\alpha$}(1.5,-4);
     \draw[dotted] (1.5,-4) -- (1.9,-4.25);
     \draw[postaction={decorate}] (em) -- node[fill=none, draw=none,left] {$\alpha$~~~}(2,-3);
      \draw[dotted] (2,-3) -- (2.4,-2.6);

 \end{tikzpicture}
  \caption{The flow problem $(H,S,T,c_{e})$ as described in the reduction
  with $e_i = e$.
  }\label{fig:klsparse_to_flow} 
\end{figure}
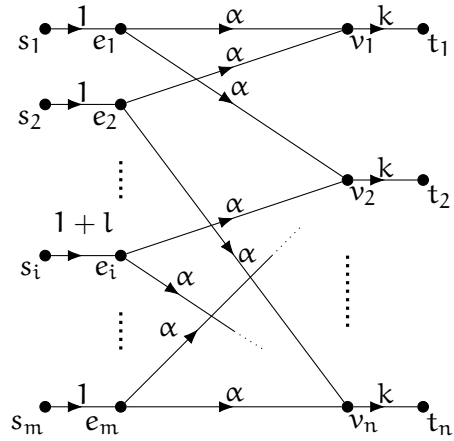

We argue that $G$ is $(k,l)$-sparse iff the flow problems $(H,S,T,c_e)$ have max flow $m+l$, for all $e \in E$.
By the \emph{max-flow min-cut theorem}, 
we can argue via $(S,T)$-cuts in~$H$.
A set~$C \subseteq V_H$ is a $(S,T)$-\emph{cut},
if $S \subseteq C$ and $C \cap T = \emptyset$.
The \emph{capacity of cut}~$C$ for capacity function~$c_e$ is
\begin{equation}
 c_e (C) = \sum_{u \in C,\,v \in V_H\setminus C} c_e(u,v).   
\end{equation}
We will show that
\[
G \text{ not } (k,l)\text{-sparse} \iff
\exists e \in E~~ \exists C\subseteq V_H ~(S,T)\text{-cut}~~~  c_e(C) < m+l.
\]
\ignore{
$G$ is not $(k,l)$-sparse iff $\exists e \in E$  the flow graph $(H,S,T,c_e)$ admits an $(S,T)$-cut~$C$ in~$H$ of capacity strictly smaller then $m+l$, 
}

Assume first that $G$ is not $(k,l)$-sparse, i.e. there is a subset $X \subseteq V$ such that $|X| \geq 2$ and $|E_G[X]| > k|X| - l$.
Choose an edge~$e \in E[X]$ and consider the cut $C = S \cup E_G[X] \cup X \subseteq V_H$ in~$H$ for capacity function~$c_e$.
The cut edges for~$C$ are the edges from~$S$ to $E \setminus E_G[X]$ of capacity~$1$
and the edges from~$X$ to~$T$ of capacity~$k$. 
Hence, we have
\begin{align}
c_e(C) &=  m - |E_G[X]| + k|X|\\
&< k|X| + m - (k|X| - l)\\
&= m + l.
\end{align}

For the backward direction, 
assume that for an edge $e \in E$ there is a flow problem $(H,S,T,c_e)$ with a cut~$C$, where $c_e(C) < m + l$.
We list some properties of cut~$C$.
\begin{itemize}
\item 
$C \not= S$, because $c_e(S) = m+l$.
\item 
$C$ does not cut through edges from~$E_2$
going from~$E$ to~$V$ as they all have capacity $\alpha = m + l + 1$ and thus, the cut would be too large
already from one such edge.

\ignore{
\item 
$C \not= S \cup E \cup V$, because $c_e(S \cup E \cup V) = kn \geq m+l$, since $G$

Also, the cut $C$ cannot only cut through edges from $S$ to $E$ or from $V$ to $T$, otherwise its value would always be $m + l$.
}
\end{itemize}

Therefore we have $C = S \cup F \cup X$ for some non-empty sets $F \subseteq E$ and $X \subseteq V$, 
such that $X'= V_G(F) \subseteq X$,
where $V_G(F)$ are the endpoints of the edges in~$F$.
The latter point is because of the second item above.
Note that $|X'| \geq 2$ since $|F| \geq 1$.

\ignore{
i.e. $C$ cuts through the edges from $S$ to $E\setminus F$ and from $X$ to $T$.
We will show that the subset of vertices $X'$ in $G$ incident to edges in $F$,
\[
X' := \{v : v \in a \text{ for some } a \in F\},
\]
violates the sparsity condition in $G$.
Note that $|X'| \geq 2$ since $|F| \geq 1$.
Moreover, observe that $X' \subseteq X$, otherwise the cut $C$ would cut trough an edge from $E \setminus F$ to $V$ with capacity $\alpha = m + l + 1$ and thus the value of $C$ would be larger than $m + l$.
}

We have two cases depending on whether the edge $e$ is in $F$ or not.
\paragraph{Case 1.} $e \in F$. 
Then we have
\begin{align}
c_e(C) &= |E \setminus F| + k|X| \\
&= m - |F|  + k|X|.
\end{align}
Hence,
$c_e(C) < m+l$ is equivalent to $|F| > k|X| -l$.
It follows that~$G$ is not $(k,l)$-sparse
because
\begin{align}
|E_G[X']| \geq |F| > k|X| -l \geq k|X'| - l.
\end{align}
The last inequality follows  because $X' \subseteq X$.
\paragraph{Case 2.}
$e \not\in F$.
Then we have
$c_e(C) = m - |F| +l + k|X|$.
Similarly as in the first case, we conclude that
$|E_G[X']| > k|X'| \geq k|X'| - l$.
\end{proof}

When $G$ has bounded genus, MSMS max flow is also in $\NC$~\cite{PS18}.
Given a graph~$G$, the procedure of Sankowski~\cite{PS18} first reduces
the max flow problem to finding the maximum weight of an \ffact{} in a weighted 
graph $G'$. This in turn is reduced to the problem of finding the
maximum weight perfect matching in a weighted graph~$G''$. 
Both steps of the reduction preserve planarity and the weights involved are polynomially bounded. 
It is easy to see that they also preserve the genus of the graph as the
constructions of $G',G''$ involve replacing vertices with small
gadgets, and can be done for any surface embedded graph without
introducing edge crossings. 
Therefore, the problem reduces to that of finding the maximum weight of a
perfect matching in a graph of bounded genus. This can be done in $\NC$~\cite{Anari2020}.

Moreover, MSMS max flow can also be solved in $\NC$ for graphs with bounded treewidth using an algorithm by Hagerup, Katajainen, Nishimura and Ragde~\cite{HAGERUP1998366}.

\begin{corollary}
For $0 \leq l \leq 2k - 1$, deciding $(k,l)$-sparsity is in $\NC$ for graphs with bounded genus,
in particular planar graphs,
and for graphs with bounded treewidth.
\end{corollary}

\end{document}